\documentclass[12pt]{amsart}
\usepackage[noadjust]{cite}
\usepackage[final]{graphicx}
\usepackage[lmargin=1in,rmargin=1in,tmargin=1in,bmargin=1in,letterpaper]{geometry}
\usepackage{booktabs}
\usepackage{url}
\usepackage[T1]{fontenc}
\usepackage{lmodern}
\usepackage[utf8]{inputenc}
\usepackage[bookmarks=true,%
    colorlinks=true,%
    linkcolor=blue,%
    citecolor=blue,%
    filecolor=blue,%
    menucolor=blue,%
    urlcolor=blue]{hyperref}
\hypersetup{pdftitle={On affine rigidity}}
\hypersetup{pdfauthor={Steven J. Gortler, Craig Gotsman, Ligang Liu,
    and Dylan P. Thurston}}
\usepackage{ifpdf}

\usepackage{caption}
\usepackage{hyphenat}
\usepackage{tikz}
\usetikzlibrary{matrix,arrows}
\tikzset{cdarrow/.style={auto,
    execute at begin node=$\scriptstyle,execute at end node=$}}

\newread\testin

\graphicspath{{draws/}}

\ifpdf
  \let\textalt\texorpdfstring
\else
  \newcommand{\textalt}[2]{#1}
\fi

\newcommand{\RR}{\mathbb R}

\newcommand{\NN}{\mathbb N}
\newcommand{\QQ}{\mathbb Q}

\newcommand{\EE}{\mathbb E}

\newcommand{\abs}[1]{{\lvert #1 \rvert}}

\DeclareMathOperator{\Eucl}{Eucl}

\DeclareMathOperator{\Aff}{Aff}

\DeclareMathOperator{\rank}{rank}

\DeclareMathOperator{\adj}{adj} 

\theoremstyle{plain}
\newtheorem{theorem}{Theorem}
\newtheorem{proposition}{Proposition}[section]
\newtheorem{lemma}[proposition]{Lemma}
\newtheorem{corollary}[proposition]{Corollary}

\theoremstyle{definition}
\newtheorem{definition}[proposition]{Definition}

\theoremstyle{remark}

\newtheorem{remark}[proposition]{Remark}

\newcommand{\Edges}{{\mathcal E}}
\newcommand{\Verts}{{\mathcal V}}

\newcommand{\cM}{\mathcal{M}}

\newcommand{\hC}{\hat{C}}
\newcommand{\hb}{\hat{b}}
\newcommand{\tm}{\tilde{m}}
\newcommand{\tM}{\tilde{M}}

\newcommand{\prho}{p}
\newcommand{\psigma}{q}

\begin{document}
\title{On affine rigidity}

\author[Gortler]{Steven J. Gortler}
\author[Gotsman]{Craig Gotsman}
\thanks{SJG and CG were partially supported by United
  States--Israel Binational Science Foundation grant 2006089.}
\author[Liu]{Ligang Liu}
\thanks{LL was partially supported 
by the National Natural Science Foundation of China (61070071).}
\author[Thurston]{Dylan P. Thurston}
\thanks {DPT was supported by the Mathematical
  Sciences Research Institute and a Sloan Research Fellowship.}

\begin{abstract}

We define the notion of
affine rigidity of a
hypergraph and  prove a variety of fundamental results for this 
notion.
First, we
show that affine rigidity can be determined by the rank of a 
specific matrix which implies that affine rigidity
is a generic property of the hypergraph.
Then we prove
that if 
a graph is is $(d+1)$-vertex-connected, then it must be
``generically neighborhood affinely rigid'' in $d$-dimensional
space.
This implies that if a graph is $(d+1)$-vertex-connected then 
any generic framework of its squared graph must be
universally rigid.

Our results, and affine rigidity more generally, have natural applications
in point registration and localization, as well as connections
to manifold learning.
\end{abstract}

\date{\today}


\maketitle


\section{Introduction}

Suppose one has a number of overlapping ``scans'' of a set of points in 
some space, and that the corresponding points shared between scans have
been identified. One naturally may want to register these scans and
merge them together into a single 
configuration~\cite[inter alia]{turk1994}. 
Such a merging
problem is called a \emph{realization problem}. 
The study of the uniqueness of the solutions to such  realization problems
is known as  \emph{rigidity}.

We model the combinatorics of this problem using a hypergraph
$\Theta$, with vertices representing the points, and hyperedges representing
the sets of points in each scan. The geometry of the problem is 
modeled with a configuration $\prho$, associating each vertex with
 a point in 
space. 

One natural setting is the Euclidean setting, where
the scans are known to be related by Euclidean transforms.
In this case it is sufficient to study just the case of a graph,
where we think of each edge as 
its own scan with only $2$ points.
Unfortunately,
many of the Euclidean problems are NP-HARD~\cite{Saxe79:EmbedGraphsNP}.
In this paper, we study what happens when one relaxes the problem to the
affine setting, that is, one assumes that the scans are 
known to be related by affine transforms. Under this relaxation,
much of the analysis reduces to linear algebra, and uniqueness questions
reduce to rank calculations.
We prove
a variety of fundamental results about this type of rigidity and also
place it in the context of other rigidity classes such as global rigidity
and universal rigidity.

We also specifically investigate the case of hypergraphs $\Theta$
that arises
by starting with an input graph $\Gamma$, and considering each one-ring
(a vertex and its neighbors)
in $\Gamma$ as a hyperedge in $\Theta$. 
We call such a hypergraph the neighborhood hypergraph of $\Gamma$.
Such neighborhood hypergraphs naturally arise
when studying molecules~\cite{lee:PhDThesis:2008}, when applying a divide and conquer approach
to sensor network localization~\cite{singer2008remark} and in machine 
learning~\cite{roweis2000nonlinear}.

\subsection{Summary of Results}
We start by describing 
how 
affine rigidity in $\RR^d$ is fully characterized by the kernel size of
one of its associated ``affinity matrices''.
(This result was first shown by Zha and Zhang ~\cite{ZZ09}.)
We show how this 
implies a number of interesting corollaries including the fact
that
affine rigidity is a generic property. That is,
given a hypergraph~$\Theta$ and dimension~$d$, either all generic
embeddings of $\Theta$ are 
affinely rigid in $\RR^d$ or all
generic embeddings are  affinely flexible in $\RR^d$. The specific
geometric positions of the vertices are  irrelevant to this
property, as long as they are in sufficiently general position.
Thus we call such a hypergraph  \emph{generically affinely
rigid} in $\RR^d$.

Next we relate affine rigidity in $\RR^d$
to the 
related notion of universal Euclidean rigidity.
A framework is universally Euclideanly rigid
if it is rigid (in the Euclidean sense) 
in \emph{any} dimension.
In this context, we prove that affine rigidity in $\RR^d$ implies 
universal Euclidean rigidity.

We then prove the following sufficiency
result: if a graph $\Gamma$ is $d+1$
(vertex) connected, then its neighborhood hypergraph is  
generically 
 affinely rigid in $\RR^d$; 
 alternatively, we say that the graph $\Gamma$ itself is 
generically neighborhood affinely rigid in $\RR^d$.
In particular we will show that almost every
\emph{non-symmetric equilibrium stress matrix}
for any generic embedding of $\Gamma$ in $\RR^d$
will have co-rank $d+1$ (i.e., rank $v-d-1$). 

Putting these two results together, we show that if a graph is 
$d+1$ connected, then any generic embedding of its square graph
in $\RR^d$ is universally rigid. This result is interesting, as very few
families of graphs have been proven to be generically universally rigid.

We give examples showing that many of the implications
proved in this paper do not reverse.

Finally we discuss some of the motivating applications.

The main properties
of frameworks of graphs and their implications 
proven in this paper are summarized
below.

\begin{tikzpicture}
  \matrix(spaces) [matrix of math nodes, column sep=0.5cm, row sep=1cm]
  {
  \textbf{GGR} & \textbf{DP1C} & \textbf{GNSESM} & 
\textbf{GNAR}  & \textbf{GNUR} & \textbf{GNGR}  \\
  };
  \draw[->] (spaces-1-1) -- (spaces-1-2);
  \draw[->] (spaces-1-2) -- (spaces-1-3);
  \draw[->] (spaces-1-3) -- (spaces-1-4);
  \draw[->] (spaces-1-4) -- (spaces-1-5);
  \draw[->] (spaces-1-5) -- (spaces-1-6);
\end{tikzpicture}

\centerline{
\begin{tabular}{rl}
  \toprule
  Property & Graph \ldots
  \\ \midrule
  \textbf{GGR} & is generically  globally rigid in $\RR^d$\\
  \textbf{DP1C} & is $d+1$ connected (\cite{Hendrickson92:ConditionsUniqueGraph})\\
  \textbf{GNSESM} & generically has  non-symmetric equilibrium stress matrix 
    of rank \\
    &\qquad $v-d-1$ (Proposition~\ref{prop:nss}) \\
  \textbf{GNAR} & is generically neighborhood affine rigid in $\RR^d$
(Proposition~\ref{prop:ar})
\\
  \textbf{GNUR} & is generically neighborhood universally rigid in $\RR^d$
(Corollary~\ref{cor:ur})
\\
  \textbf{GNGR} & is generically neighborhood globally  rigid in $\RR^d$ (by definition)
\\ \bottomrule 
\end{tabular}}


\section{A rigidity zoo}
\label{sec:rigidity-zoo}

In this paper we will consider several different rigidity theories.
They all fit in to a unifying framework, which we now explain.

Most generally, rigidity questions
(of any type) ask if  all of the geometric
information about a set of points is  determined by
information from small subsets.  In the usual Euclidean rigidity
problem, we measure the distances between pairs of points.  However, in
other cases it is not enough to consider pairs of points for the small
subsets; as such, we need to consider hypergraphs rather than just
graphs.

\begin{definition}
  A \emph{hypergraph}~$\Theta$ is a set of $v$ vertices $\Verts(\Theta)$
  and $h$ hyperedges~$\Edges(\Theta)$, where $\Edges(\Theta)$ is a set of
  subsets of $\Verts(\Theta)$.  We will typically write just $\Verts$
  or $\Edges$, dropping
  the hypergraph~$\Theta$ from the notation.
\end{definition}

There are natural ways to pass from a hypergraph back and forth to a graph.

\begin{definition}
Given a hypergraph $\Theta$, define its 
\emph{body graph} $B(\Theta)$
as follows. For each vertex in $\Theta$, we have a vertex in
$B(\Theta)$. For each hyperedge $h$ in $\Theta$
and each pair of vertices in $h$
we have an edge in $B(\Theta)$.
\end{definition}

See Figure~\ref{fig:hyper-bodygraph} for an example.
\begin{figure}
  \centering
  \includegraphics{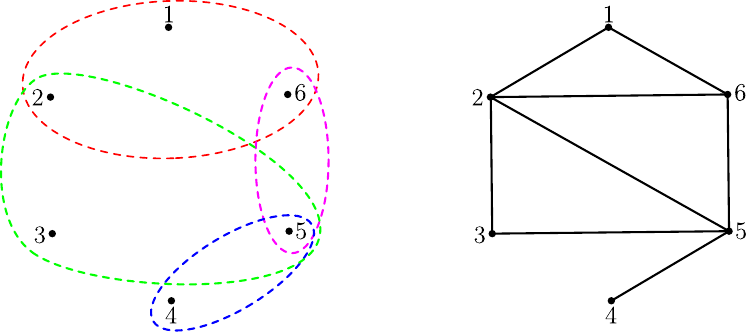}
  \caption{Left: A hypergraph with 6 vertices and 4 hyperedges. Each
    hyperedge is represented by a dotted ellipse enclosing a set of
    vertices. The hyperedges are: red $\{1,2,6\}$, green $\{2,3,5\}$,
    purple $\{5,6\}$, blue $\{4,5\}$.
    Right: The body graph of the hypergraph shown in the left.}
  \label{fig:hyper-bodygraph}
\end{figure}

\begin{definition}
  Given a graph $\Gamma$, define its 
\emph{neighborhood  hypergraph},
written as $N(\Gamma)$ as follows.
  For each vertex in $\Gamma$, we have an associated vertex in $N(\Gamma)$.
  For each vertex in $\Gamma$ we have a hyperedge in $N(\Gamma)$ consisting
of that vertex and its neighbors in $\Gamma$.
\end{definition}

\begin{definition}
Given a graph~$\Gamma$, its 
\emph{squared graph} $\Gamma^2$ is obtained
by adding to $\Gamma$ an edge between two vertices $i$ and $j$ if
$i$ and $j$ share some neighbor vertex $k$.
\end{definition}

\begin{lemma}\label{lem:square-nbhd}
  For any graph~$\Gamma$, $B(N(\Gamma)) = \Gamma^2$.
\end{lemma}

\begin{proof}
  Immediate from the definitions.  (See Figure~\ref{fig:neighbor-hyper}
  for an example.)
\end{proof}

\begin{figure}
  \centering
  \includegraphics{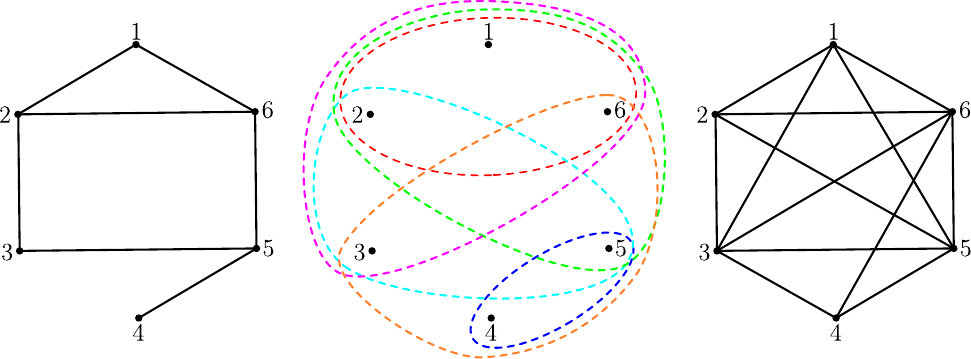}
  \caption{Left: A graph with 6 vertices and 7 edges. 
    Middle: Its neighborhood hypergraph with 6 vertices and 6
    hyperedges: red $\{1,2,6\}$, purple $\{1,2,3,6\}$, light blue
    $\{2,3,5\}$, dark blue $\{4,5\}$, orange $\{3,4,5,6\}$, green
    $\{1,2,5,6\}$.
    Right: The body graph of the hypergraph in the middle. It is also
    the squared graph of the graph in the left.}
  \label{fig:neighbor-hyper}
\end{figure}

\begin{definition}
  A \emph{$k$-hypergraph}~$\Theta$ is a hypergraph where each hyperedge has
  exactly $k$ vertices.
For any $k\in \NN$ and hypergraph~$\Theta$, let $B_k(\Theta)$ be the
$k$-hypergraph whose hyperedges are all the subsets $S$
of size~$k$ of vertices that are contained together in
at least one hyperedge
of $\Theta$:
\[
\Edges(B_k(\Theta)) = \{\, S \mid \exists h \in \Edges(\Theta),
  S \subset h, \abs{S} = k\,\}.
\]
For a vertex set~$S$, the \emph{complete $k$-hypergraph on $S$},
written $K_k(S)$, is the $k$-hypergraph whose hyperedges are all
$\binom{\abs{S}}{k}$ subsets of $S$ of size~$k$.
\end{definition}

For instance, a $2$-hypergraph is a graph, and $B_2(\Theta)$ is just the
body graph~$B(\Theta)$.

\begin{definition}
  A \emph{configuration} of the vertices $\Verts(\Theta)$ of a
  hypergraph in $\RR^d$ is 
  a mapping~$p$ from $\Verts(\Theta)$ to~$\RR^d$.
Let $C^d(\Verts)$ 
be the space of configurations in $\RR^d$.  
For $p \in C^d(\Verts)$ and $u
  \in \Verts(\Theta)$, 
we write $p(u) \in \RR^d$ for the image of $u$ under~$p$.

A \emph{framework} $\rho = (p,\Theta)$ of a hypergraph is the pair of a 
hypergraph and a configuration
of its vertices.  $C^d(\Theta)$ is the space of frameworks $(p,
\Theta)$ with hypergraph $\Theta$ and $p \in C^d(\Verts(\Theta))$.
We may also
write $\rho(u)$ for $p(u)$ where $\rho = (p,\Theta)$ is a framework of the
configuration~$p$.
\end{definition}

A framework of a hypergraph has also been called a body-and-joint 
framework~\cite{whiteley1989matroid}
and a body-and-multipin
framework~\cite{lee:PhDThesis:2008}.

\begin{definition}
  Let $M$ be a monoid  acting on $\RR^d$, such as 
  $\Eucl(d)$, the Euclidean isometries of $\RR^d$.
(We study monoids instead of groups since we don't want to 
restrict ourselves to always having inverses. In particular in
the case of affine rigidity, we wish to allow singular affine 
transforms as well).
The framework $\rho \in C^d(\Theta)$ 
is 
\emph{$M$-preequivalent}
to the framework $\sigma \in C^d(\Theta)$ 
if for each hyperedge $h \in \Edges(\Theta)$, the positions of the
vertices in $\rho$ can be mapped to their positions in~$\sigma$ by
an element of~$M$ depending only on~$h$.  That is, for each $h \in
\Edges$, there is a
$g_h \in M$ so that for each $u \in
h$, we have $g_h(\rho(u)) = \sigma(u)$.

The configuration $p \in C^d(\Verts)$ 
is 
\emph{$M$-precongruent}
to the configuration  $q \in C^d(\Verts)$ 
if the positions of
all the vertices in $p$ can be mapped to their positions in
$q$ by a single element $g \in M$ (not depending on~$h$).
When the configuration $p$ is $M$-precongruent to $q$, we also say
that the framework $(p,\Theta)$
is $M$-precongruent to $(q,\Theta)$.

A framework
$\rho  \in C^d(\Theta)$ is
\emph{globally $M$-rigid} if for any
other framework $\sigma  \in C^d(\Theta)$ to which 
$\rho$ 
is $M$-preequivalent, 
we also have that 
$\rho$ is 
$M$-precongruent to
 $\sigma$.
Otherwise we say
that $\rho$ is globally $M$-flexible in $\RR^d$.

Similarly, a framework~$\rho \in C^d(\Theta)$ is
\emph{locally $M$-rigid} in $\RR^d$ if there is a small
neighborhood~$U$ of $\rho$ in $C^d(\Theta)$ so that for any
 $\sigma  \in U$ to which 
$\rho$,
is $M$-preequivalent,
we also have that 
$\rho$ is $M$-precongruent to
$\sigma$.
Otherwise we say
that $\rho$ is locally $M$-flexible in~$\RR^d$.
\end{definition}

\begin{remark}
When there are non-invertible elements of
$M$, then neither 
$M$-preequivalence nor 
$M$-precongruence
is a symmetric relation.
When $M$ is a group, then $M$-preequivalence is a symmetric relation
and can be called $M$-equivalence, and
likewise 
$M$-precongruence can be called
$M$-congruence. 
\end{remark}

A related notion of group based rigidity has been explored in the
computer aided design literature~\cite{SchreckM06}.

In this paper, we are mainly
concerned with the cases when $M$ is either $\Eucl(d)$ or $\Aff(d)$, 
the set  of all (including singular)
affine linear maps of $\RR^d$, in which case we speak about
\emph{Euclidean} or \emph{affine} rigidity, respectively.  But there
are other interesting
possibilities, like projective transformations.
Another interesting case is when $M$ consists of the dilations and
translations of $\RR^d$ (with no rotations); this gives the theory
of parallel-line redrawings~\cite{whiteley1989matroid}.

In this terminology, 
Euclidean rigidity is the default: if $M$ is not
specified, it is the Euclidean group.
In much of the rigidity literature, 
local rigidity is the default, and the qualifier ``local'' is
dropped. However, in this paper this distinction is important and we
will write ``local'' or ``global'' when the distinction is meaningful.

\begin{lemma}\label{lem:hyper-er}
  A framework $(p, \Theta)$ is locally (resp.\ globally) Euclideanly rigid
  iff the body framework $(p, B(\Theta))$ is locally (resp.\ globally)
  Euclideanly rigid.
\end{lemma}

\begin{proof}
  This easily follows from the fact that, for each hyperedge $h \in
  \Edges(\Theta)$, the
  complete graph on $\abs{h}$ vertices is globally rigid.
\end{proof}

Thus we only need to consider Euclidean rigidity for frameworks of
graphs, not hypergraphs.

In the next section 
(Corollary~\ref{cor:glocal}) 
we will see that a framework is locally affinely
rigid iff it is globally affinely rigid.  Thus we can
drop the local/global distinction for affine rigidity.

In the following definition,
for $d < d'$, we view
$C^d(\Verts)$ as contained in $C^{d'}(\Verts)$ 
by the inclusion of $\RR^d$ as the
first $d$ coordinates of $\RR^{d'}$.  

\begin{definition}
\label{def:groups}
  Let $M$ be a family of monoids~$M_d$ acting on $\RR^d$, so that for
  $d < d'$, $M_d$ is the submonoid of $M_{d'}$ that fixes $\RR^d$ as a
  subset of $\RR^{d'}$.
  A framework $\rho \in C^d(\Theta)$ is \emph{universally} locally (resp.\
  globally) $M$-rigid if it is locally (resp.\ globally) $M$-rigid as
  a framework in $C^{d'}(\Theta)$ for all $d' \geq d$.
\end{definition}

Note that universal rigidity of any sort implies rigidity of the same
sort.

\begin{lemma}\label{lem:ur-local-global}
  A framework $p \in C^d(\Theta)$ is universally globally Euclideanly rigid iff
  it is universally locally Euclideanly rigid. 
\end{lemma}

\begin{proof}
  For any two equivalent frameworks $\rho$ 
in $C^d(\Theta)$ and $\rho'$ in
  $C^{d'}(\Theta)$,  Bezdek and
  Connelly~\cite{BC02:PushingDisks} show how to construct
 an explicit flex
  between $\rho$ and $\rho'$ in $C^{d+d'}(\Theta)$.  Thus, if
  $\rho$ is a $d$-dimensional framework with a equivalent but
  non-congruent framework in $d'$ dimensions, then their constructed flex shows
  that $\rho$ is not locally rigid in $\RR^{d+d'}$.
\end{proof}

Thus we can also
drop the local/global distinction in the case of universal Euclidean rigidity.

\begin{definition}
  A framework $(p, \Gamma)$ of the graph $\Gamma$
is \emph{neighborhood} rigid (of any of
  the sorts above) if the corresponding framework
  $(p, N(\Gamma))$ of the neighborhood hypergraph is rigid (of the same
  sort).
\end{definition}

For instance, Lemmas~\ref{lem:hyper-er} and~\ref{lem:square-nbhd} tell
us that neighborhood
Euclidean rigidity of $(\rho, \Gamma)$ is equivalent to the
Euclidean rigidity
of $(\rho, \Gamma^2)$.

\begin{remark}
Related to universal Euclidean rigidity is the notion of
\emph{dimensional rigidity}~\cite{alf}. A framework (locally rigid or not)
in $\RR^d$
is called dimensionally rigid if there is no (Euclidean) equivalent framework
with an affine span of dimension strictly greater than $d$.

Another related notion is that of \emph{$d$-realizability}~\cite{belk}. 
A graph is $d$-realizable if any framework of the graph, in any dimension,
has a
(Euclidean) equivalent framework with an affine span of dimension
$d$ or less. 

Presumably one could extend these notions to arbitrary monoids as well
but we
will not pursue these in this paper. 
\end{remark}

\begin{definition}
  \label{def:generic}
  A configuration~$p$
in $C^d(\Verts)$
is \emph{generic} if the 
coordinates do
  not satisfy any non-zero algebraic equation with rational
  coefficients.  A framework is generic if its configuration is
  generic.
\end{definition}

A property is \emph{generic} in $\RR^d$
if, for every (hyper)graph, either all generic frameworks in
$C^d(\Theta)$
have the property
or none do.  For instance, local and global Euclidean rigidity in
$\RR^d$ are both generic properties of graphs and therefore for
hypergraphs as well~\cite{AR78:RigidityGraphs,GHT10}.  
For any property~$P$ (generic or not) of frameworks,
a (hyper)graph~$\Theta$ is \emph{generically~$P$} 
in $\RR^d$ if every
generic framework in $C^d(\Theta)$ has~$P$.  
(For a non-generic property like universal Euclidean rigidity, there
are (hyper)graphs that are neither
generically~$P$ or generically not~$P$.)

Thus, for any framework, we may talk about
\begin{quote}
  (generic/$\emptyset$) (universal/$\emptyset$) (local/global)
  (Euclidean/affine) rigidity.
\end{quote}
where by $\emptyset$ we mean that this term has been dropped.


\section{Affine Rigidity in \textalt{$\RR^d$}{R\textasciicircum d}}

We now move on the main focus of this work, affine rigidity, as
defined in the previous section.
Though the definitions start from a different point of view, this
notion of affine rigidity is, in fact,
identical to the one defined
by Zha and Zhang~\cite{ZZ09}
and the concept is also informally
mentioned by Brand~\cite{brand2004subspaces}. Additionally,
Theorem~\ref{thm:ranktest} below is
essentially equivalent to~\cite[Theorem 5.2]{ZZ09}.

Our contribution here, described by the corollaries, is 
showing how affine rigidity fits in to the general scheme of rigidity
problems.

\begin{lemma}\label{lem:complete-aff-rigid}
  Any framework of a complete $(d+2)$-hypergraph
in $\RR^d$
  is affinely rigid.
\end{lemma}

\begin{proof}
Let $q \in C^{d}(\Verts)$ 
be a configuration with 
such that 
$(p,\Theta)$ is affinely preequivalent in $\RR^{d}$ to 
$(q,\Theta)$. 

Let $c \leq d$ be the dimension of the affine span, $S$, of the 
configuration $p$. Select $c+1$ vertices whose affine span in $p$
is $S$.
Let $A_0$ be an affine transform that maps these vertices 
from their positions in $p$ to their positions in $q$. 
The action of $A_0$
on the space  $S$ (and thus all of the vertices $p$)
is uniquely determined by these selected vertices.

For any vertex $v$, there must be a hyperedge $h_v$ that includes $v$ and the
selected vertices. 
Let $A_{h_v}$ be an affine transform that maps these vertices 
from their positions in $p$ to $q$ (which  must exist by affine
preequivalence).
The action of $A_{h_v}$
on the space  $S$ (and thus all of the vertices $p$) 
is uniquely determined by the selected vertices, and thus must agree
with that of $A_0$. Thus for all $v$, we see that their positions in $q$ are 
obtained from the positions in $p$ through $A_0$
Thus $p$ is affinely precongruent to $q$.
\end{proof}

\begin{proposition}
\label{prop:khype}
  A framework $(p,\Theta)$ in general position is affinely locally
  (resp.\ globally)
  rigid iff the associated framework
  $(p, B_{d+2}(\Theta))$ is affinely locally (resp.\ globally) rigid.
\end{proposition}

(Compare Lemma~\ref{lem:hyper-er}.)

\begin{proof}
First
consider a hyperedge with less than $d+2$ vertices in general position.
Using  
a $d$-dimensional affine transform, we can move these vertices to
any other configuration 
in $\RR^d$.
Therefore this hyperedge does not
affect affine preequivalence and may
be dropped without affecting
affine rigidity.
Next consider a hyperedge with $k$ vertices, with $k > d+2$.
By Lemma~\ref{lem:complete-aff-rigid},
one can replace this hyperedge with
$\binom{k}{d+2}$ hyperedges corresponding to all subsets of $d+2$
vertices. The framework of the new hypergraph
will be affinely rigid in $\RR^d$ iff the original one is.
\end{proof}

\begin{definition}\label{def:affinity-matrix}
An \emph{affinity matrix} of a framework $(p,\Theta)$ 
in $C^d(\Theta)$
is a matrix with $v$ columns such that 
each row encodes some affine 
relation between the coordinates of 
the vertices in a hyperedge of 
$(p,\Theta)$ as a homogeneous linear equation in the following sense.
The only non-zero entries in a row correspond to vertices in
some hyperedge, the sum of the entries in a row is $0$, and 
each of the coordinates of~$p$, thought of as a vector of
length~$v$,
is in the kernel of the matrix.

An affinity matrix is \emph{strong} if it encodes all of the 
affinely independent relations in every hyperedge of
$(p,\Theta)$. 
\end{definition}

\begin{lemma}
\label{lem:affinity-equiv}
If the framework 
$(p, \Theta)$ 
is affinely preequivalent to 
$(q,\Theta)$ 
then
the
coordinates of $q$ are in the kernel of any affinity matrix for
$(p, \Theta)$.
Additionally, the converse is true  if the affinity matrix is strong.
\end{lemma}
\begin{proof}
  Clear from the definitions.
\end{proof}

The kernel of an affinity matrix of a framework
$(p,\Theta)\in C^d(\Theta)$ 
always contains the subspace of $\RR^v$ 
spanned by the coordinates of~$p$ along each axis and the
vector~$\vec 1$ of all~$1$'s. This corresponds to the fact that any
$p$ is preequivalent to any of its affine images.
If $p$ is a proper $d$-dimensional configuration
(with full $d$-dimensional affine span),
these vectors are independent and span a
$(d+1)$-dimensional space.  In particular,
a generic framework of a hypergraph
with at least $d+1$ vertices in $\RR^d$
is proper, so for such frameworks the corank of any of its affinity 
matrices must be no less than $d+1$.

The rank of strong affinity matrices fully characterize affine rigidity.

\begin{theorem}
\label{thm:ranktest} 
Let $\Theta$ be a hypergraph with at least $d+1$ vertices.
Let $(p,\Theta)$
be any proper, $d$-dimensional framework
and let $M$ be any strong affinity matrix for $(p,\Theta)$.
Then 
$(p,\Theta)$
is affinely rigid in $\RR^d$ iff 
$\dim(\ker(M))= d+1$.
\end{theorem}
\begin{proof}
By Lemma~\ref{lem:affinity-equiv}, for any other 
configuration $q$ in $C^d(\Verts)$ such that 
$(p,\Theta)$ is  affinely preequivalent to 
$(q,\Theta)$, 
the coordinates of
$q$ 
must be in the kernel of $M$. 
When
$\dim(\ker(M))=d+1$, the kernel of $M$ 
contains only one-dimensional
projections of $p$ and the all-ones vector.
Thus when
$(p,\Theta)$ 
 is  affinely preequivalent to 
$(q,\Theta)$, we have that 
$(p,\Theta)$ 
must in fact be affinely   precongruent  to 
$(q,\Theta)$.

Conversely, if the corank is higher, then the kernel must 
contain an ``extra'' vector that is not a one-dimensional
projection of $p$. Adding any amount of 
this vector to one of the coordinates
of $p$ must, by 
Lemma~\ref{lem:affinity-equiv}, 
produce a $q$ such that 
$(p,\Theta)$ is  affinely preequivalent to 
$(q,\Theta)$ 
but not  precongruent  to it.
\end{proof}

It is easy now to prove the following corollaries.

\begin{corollary}
\label{cor:glocal}
If $(p,\Theta)$ 
is affinely globally flexible in $\RR^d$ 
then it is affinely locally flexible.
\end{corollary}
\begin{proof}
From the proof of Theorem~\ref{thm:ranktest}, when
$(p,\Theta)$ 
is affinely globally flexible in $\RR^d$ there is an extra vector
$\delta$ in the kernel of a strong affinity matrix, and we can add any
multiple of $\delta$ to one of the coordinates of $p$ to get a framework
to which $(p,\Theta)$ is 
affinely preequivalent but not precongruent.
\end{proof}

\begin{remark}
In fact, 
if $(p,\Theta)$ 
is affinely preequivalent 
to $(q,\Theta)$,
there is a continuous path
of frameworks in $C^d(\Theta)$ to which
$(p,\Theta)$ 
is affinely preequivalent,
namely $((1-t)p + tq, \Theta)$.
\end{remark}

\begin{corollary}
\label{cor:uar}
A framework $(p,\Theta) \in C^d(\Theta)$
is affinely rigid in $\RR^d$ iff it
is affinely rigid when considered as a (degenerate) framework
in $\RR^{d'}$ for $d' \geq d$.
\end{corollary}
\begin{proof}
Follows from the proof of Theorem~\ref{thm:ranktest}.
\end{proof}

Thus there is no distinct notion of 
``universal''  affine rigidity.

\begin{corollary}
\label{cor:arGen}
Affine rigidity in $\RR^d$
is a generic property of a hypergraph.
\end{corollary}
\begin{proof}
The condition that $M$ is an affinity matrix for $(p,\Theta)$
is linear in the entries in~$M$.
The corollary then follows from 
Proposition~\ref{prop:genrank}.
\end{proof}

Though we will not pursue the details here, one can use the 
concept of an affinity matrix to derive
an efficient randomized discrete algorithm for testing generic
affine rigidity of a hypergraph in $\RR^d$. 
To do this, one needs to use integers
of sufficiently many bits, and do the arithmetic modulo a suitably
large prime. The details parallel 
those in the global rigidity case~\cite[Section~5]{GHT10}.

\begin{remark}
There is also an strong  connection between affine rigidity
and a problem from polyhedral scene analysis~\cite{whiteley1989matroid}.
This is most easily explained in two dimensions.
Given a framework in $\RR^2$, one can interpret each hyperedge as a  
 planar polygon drawn in $\RR^2$  (the vertex order is
not relevant here). We say that the framework is \emph{sharp}
if each vertex can be given a
third coordinate, such that, in the resulting three dimensional drawing,
each polygon remains planar, and the faces do not all lie in a single
plane. This idea is easily generalized to arbitrary dimension.

As shown in~\cite[Proposition 2.1]{whiteley1989matroid}, 
a framework is sharp iff
the rank of its strong affinity matrix is not maximal. Thus this 
notion of sharpness corresponds
exactly to affine flexibility.

More deeply, due the combinatorial characterization
of sharpness given by~\cite[Theorem 4.2]{whiteley1989matroid},  generic
affine rigidity can be tested by an efficient  
deterministic algorithm.
\end{remark}


\section{Universal Euclidean Rigidity}

We now turn to universal Euclidean rigidity.
To begin, we need the following technical definition:

\begin{definition}
  \label{def:conic}
  We say that the edge directions of a graph framework
  $(p,\Gamma) \in C^d(\Gamma)$ 
are \emph{on a conic at
  infinity} if there exists a symmetric $d$-by-$d$ matrix~$Q$ such that
for all edges $(u,v)$ of 
$\Gamma$, we have
\[ [p(u)-p(v)]^t Q [p(u)-p(v)] = 0.\]
\end{definition}

The 
edge directions  of $(p,\Gamma)$
are 
on a conic at infinity iff  there is
a continuous family of $d$-dimensional non-Euclidean
affine transforms that
preserve all of the edge lengths~\cite{Connelly05:GenericGlobalRigidity}.
This is a very degenerate situation which is 
very easy to rule out. For example, if in a hypergraph framework $(p,
\Theta)$
some  hyperedge in $\Theta$ has vertices whose positions in~$p$
affinely span $\RR^d$, then 
the 
edge directions  $(p,B(\Theta))$ cannot be 
on a conic at infinity.

\begin{theorem}
\label{thm:a2u}
If a framework $(p,\Theta)$ of a hypergraph $\Theta$
with $p \in C^d(\Verts)$ is
affinely rigid in $\RR^d$, and  
the 
edge directions of $(p,B(\Theta))$ are
not on a conic at infinity, 
then $(p,\Theta)$ 
is universally Euclidean rigid.
\end{theorem}

\begin{proof}
Let $q \in C^{d'}(\Verts)$ 
be a configuration with $d' > d$
such that 
$(p,\Theta)$ is Euclidean equivalent in $\RR^{d'}$ to 
$(q,\Theta)$. 
Then $(p,\Theta)$ 
is affinely preequivalent in $\RR^{d'}$ to 
$(q,\Theta)$. 
Since $(p,\Theta)$ is affinely rigid in $\RR^d$, 
from Corollary~\ref{cor:uar}, we have that 
$p$ is affinely precongruent to 
$q$ in $\RR^{d'}$
and 
the affine span of 
$q$ must be of dimension no larger than $d$.
Let $R(q)$ be a rotation of $q$ down to~$\RR^d$.
Then
$p$,  
is affine  precongruent in $\RR^d$ to 
$R(q)$
and 
$(p,\Theta)$
is  Euclidean  equivalent  in $\RR^d$ to 
$(R(q),\Theta)$.

Let $A$ be an affine transform such that
$A(p)=R(q)$ (which must exist due to affine precongruence).
By Euclidean equivalence, all of edge lengths
agree between
$(p,B(\Theta))$ and
$(A(p),B(\Theta))$.
If $A$ is not Euclidean, then this means that the 
edge directions of $(p,B(\Theta))$ are
on a conic at infinity, which contradicts our assumption.
Thus $A$ is Euclidean making 
$p$ and 
$R(q)$ 
Euclidean congruent in $\RR^d$.
Likewise $p$
must be Euclidean congruent to $q$ in $\RR^{d'}$, and we can conclude
that 
$(p,\Theta)$ 
is universally rigid.
\end{proof}

\begin{corollary}
\label{cor:ur}
Let $\Theta$ be a hypergraph with at least $d+2$ vertices.
If a generic framework $(p,\Theta)$ of a hypergraph $\Theta$
with $p \in C^d(\Verts)$ is
affinely rigid in $\RR^d$
then $(p,\Theta)$ 
is universally rigid.
\end{corollary}
\begin{proof}
As in the proof of Proposition~\ref{prop:khype}, 
any generic framework of a hypergraph $\Theta$ with at least $d+2$
vertices
that is affinely rigid in 
$\RR^d$, must have at least one hyperedge $h$ with at least $d+2$ vertices,
and these
vertices in~$p$ have a $d$-dimensional  affine span. Thus 
$(p,B(\Theta))$ must include
a general position  
framework of a $(d+1)$-simplex and thus cannot have edge directions
at a conic at infinity.  Then Theorem~\ref{thm:a2u} applies.
\end{proof}

There can be frameworks that are universally rigid but not
affinely rigid in $\RR^d$. (See Figure~\ref{fig:tetra}.)

\begin{figure}[t]
\center
\includegraphics{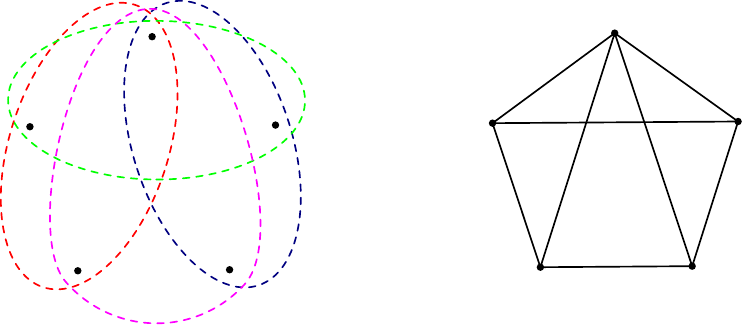}
\caption{\label{fig:tetra} The framework in $\RR^2$ of the hypergraph
  on the left
 is not affinely rigid as each hyperedge (shown as a dashed ellipse)
has only 3 vertices.  But
 this framework 
 is universally Euclidean rigid, as its body graph (right) is a Cauchy
 polygon. (A Cauchy (bar) polygon
on $v$ vertices
is a planar framework where the vertices $p(1),...,p(v)$, in order,
form a strictly convex polygon in the plane, and the edge set
consists of the edges $\{i, {i+1}\}$, $i=1,2,...,v$, and
$\i, {i+2}\}$, $i=1,...,v-2$ 
(indices modulo $v$). 
A Cauchy polygon is universally rigid~\cite{Connelly82:RigidityEnergy}).
}
\end{figure}

Corollary~\ref{cor:ur} 
can be generalized beyond the Euclidean case
to apply to much larger set of groups and monoids.\footnote{Thanks to 
Louis Theran for suggesting we look at this generality.}

\begin{theorem}
Let $M$ be a 
family of monoids  $M_d$ (as in Definition~\ref{def:groups})
with 
each $M_d$ a submonoid  of $\Aff(d)$.
Let $(p,\Theta) \in C^d(\Theta)$ 
be a framework with 
some  hyperedge $h_0$ in $\Theta$ whose vertex positions in~$p$
affinely span $\RR^d$.
If 
$(p,\Theta)$ 
is 
affinely rigid in $\RR^d$
then 
$(p,\Theta)$ 
is universally $M$-rigid.
\end{theorem}

\begin{proof}
Let $q \in C^{d'}(\Verts)$ 
be a configuration with $d' > d$
such that 
$(p,\Theta)$ is $M$-preequivalent in $\RR^{d'}$ to 
$(q,\Theta)$. 
Then $(p,\Theta)$ is affinely preequivalent in $\RR^{d'}$ to 
$(q,\Theta)$. 
Since $(q,\Theta)$ is affinely rigid in $\RR^d$, 
from Corollary~\ref{cor:uar}, we have that 
$p$ is affinely precongruent to 
$q$ in $\RR^{d'}$; there is 
an $A \in \Aff(d')$ such that $A(p)=q$.

By the  assumption of $M$-preequivalence, for each hyperedge~$h$ there is 
an element $g_h \in M_{d'}$ which maps the  vertices of $h$
from their positions in $p$ to their positions in $q$.
Since $M_{d'}$ is 
a subgroup of $\Aff(d')$, and the specific hyperedge $h_0$ 
has $d+1$ vertices in general position in the configuration $p$,
the action of $g_{h_0}$ on $\RR^d$
is fully determined by these vertices and must agree with the action
of $A$ on $\RR^d$. 
Thus $g_{h_0}(p) = A(p) = q$, making $p$ $M$-precongruent to $q$, 
and making
$(p,\Theta)$ 
universally $M$-rigid.
\end{proof}


\section{Neighborhood affine rigidity}

In this section we prove the following theorem about the 
generic neighborhood affine 
rigidity of a graph.

\begin{theorem}
\label{thm:dcon}
Let $\Gamma$ be a graph.
If $\Gamma$ is $(d+1)$-vertex-connected, then $\Gamma$
is generically 
neighborhood affinely rigid in $\RR^d$.
\end{theorem}

The strategy to prove this theorem is as follows.
First we show, using a ``rubber band'' 
construction~\cite{LLW88, Connelly82:RigidityEnergy, tutte},
that a sufficiently connected graph must
have a framework with certain nice geometric properties.
Moreover, these geometric properties are stable under generic
perturbations of the configuration.
Then we show that any such  framework must have a 
``non-symmetric equilibrium stress matrix'' of appropriate high rank.
Since the perturbed framework is generic, then any 
generic framework must have such a matrix. This matrix 
serves as a certificate
of neighborhood affine rigidity.

\begin{definition}
  \label{def:stress-matrix}
  An \emph{equilibrium stress matrix} of a framework~$(p,\Gamma)$
of a graph in $C^d(\Gamma)$
is a matrix~$\Omega$
  indexed by $\Verts\times \Verts$ so that
  \begin{enumerate}
  \item for all $u,w \in \Verts$, we have $\Omega(u,w) = \Omega(w,u)$;
  \item for all $u,w \in \Verts$ with $u \ne w$ and $\{u,w\} \not\in \Edges$,
    we have $\Omega(u,w) = 0$;
  \item for all $u \in \Verts$, we have $\sum_{w\in\Verts} \Omega(u,w) = 0$; and
  \item for all
    $u \in \Verts$, we have
    $\sum_{w\in\Verts} \Omega(u,w)p(w) = 0$.
  \end{enumerate}
  A \emph{non-symmetric equilibrium stress matrix} of a framework 
$(p,\Gamma)$ is a matrix that satisfies properties (2)--(4) above.
\end{definition}

Observe first that an equilibrium stress matrix (symmetric or not)
$\Omega$ of
$(p,\Gamma)$ is an affinity matrix
of $(p,N(\Gamma))$.
From the properties of affinity matrices,
the kernel of $\Omega$
always contains the subspace
spanned by the coordinates of~$p$ along each axis and the
vector~$\vec 1$ of all $1$'s.

\begin{definition}
We say that a framework 
of a graph in $C^d(\Gamma)$
has the \emph{convex containment} property if
\begin{enumerate}
\item the configuration of each vertex along with its neighboring vertices
has an affine span of 
  dimension~$d$, and
\item 
Almost every  vertex in the framework
is contained in the strict $d$-dimensional 
convex hull of its neighbors.
There are may be up to $d+1$, so-called, \emph{exceptional vertices}  which
do not have this property.
\end{enumerate}
\end{definition}

\begin{lemma}
\label{lem:rubber}
Let $\Gamma$ be a graph with at least $d+1$ vertices.
Suppose $\Gamma$ is $(d+1)$-connected.
Then there exists a generic framework 
$(q,\Gamma)$ in $C^d(\Gamma)$ with the
convex containment property.
\end{lemma}
\begin{proof}
Pick any $d+1$  vertices to be exceptional.
Constrain the exceptional vertices to fixed generic positions in
$\RR^d$ (at the vertices of a simplex). Associate generic positive
weights~$\omega_{ij}$
with each undirected edge $ij$. Find the ``rubber band''
configuration consistent with the constrained vertices and these
weights. Namely, find a framework  $(r,\Gamma)$ so that
each non-exceptional vertex is the weighted linear average of its
neighbors:
\[
\sum_{j \in N(i)} \omega_{ij} (r(i)-r(j))=0,
\]
where
$N(i)$ are the neighbors of vertex $i$.  This involves solving $d$
systems of linear equations, one for each component of~$r$. 
Note that the resulting configuration~$r$ may not be generic.

From~\cite{LLW88},
we know that if $\Gamma$ is $(d+1)$-connected and the constraints on
the exceptional vertices and the edge weights~$\omega$ are generic,
then no set of $d+1$ vertices in $r$ will be
contained in a $(d-1)$-dimensional affine plane, giving us the
first condition.

By construction, any non-exceptional  vertex in $(r,\Gamma)$ must be
contained in the convex hull of its neighbors. 
Again, from~\cite{LLW88}, 
the convex containment must be  strict.

Finally, we perturb each vertex in $\RR^d$
to obtain a generic configuration in $q \in C^d(\Verts)$.
By the first convex containment condition, the convex hull of the
neighbors of a vertex has non-empty interior, so a sufficiently small
perturbation will maintain both conditions.
\end{proof}

\begin{definition}
Suppose that $(q,\Gamma)$ has the convex containment property
and $\Omega$ is a non-symmetric
equilibrium stress matrix for $(q,\Gamma)$. We call a row
of~$\Omega$ \emph{non-exceptional} if its corresponding vertex 
is in the strict $d$-dimensional 
convex hull of its neighbors.
\end{definition}

\begin{lemma}
\label{lem:positivity}
Let $\Gamma$ be a graph with at least $d+1$ vertices.
Suppose $\Gamma$ is a $(d+1)$-connected graph, and we have a framework
$(q,\Gamma)$ in
$C^d(\Gamma)$ with the convex containment property.
Then there is a non-symmetric equilibrium
stress matrix $\Omega$ of $(q,\Gamma)$, such that
for every non-exceptional row $i$,
we have the following property:
If there is an edge connecting vertex $i$ and vertex $j$,
then $\Omega_{ij}$ is positive.
\end{lemma}
\begin{proof}
All vertices have $d+1$ or more neighbors. 
For each vertex~$i$, we can therefore find ``barycentric
coordinates'': non-zero edge weights~$\omega_{ij}$ on the adjoining edges so
that
\[
\sum_{j \in N(i)} \omega_{ij} (q(j) - q(i)) = 0.
\]
If $i$ is a non-exceptional vertex, 
due to the convex containment property we can choose the
$\omega_{ij}$ to be positive.  We then choose $\Omega_{ij} =
\omega_{ij}$ for $i \ne j$ and $\Omega_{ii} = -\sum_j \omega_{ij}$.
\end{proof}

\begin{remark}
This lemma is false if we require the stress matrix to be symmetric,
because this prevents us from choosing $\omega_{ij}$ and $\omega_{ji}$
independently.
\end{remark}

\begin{lemma}
\label{lem:convex-rank} 
Let $\Gamma$ be a graph with at least $d+1$ vertices.
Suppose $\Gamma$ is $(d+1)$-connected, and
we have a framework $(q,\Gamma)$ in $C^d(\Gamma)$ with the 
convex containment property.
Then there is a non-symmetric
equilibrium stress matrix $\Omega$ of $(q,\Gamma)$ with co-rank
$d+1$.
\end{lemma}

\begin{proof}

From Lemma~\ref{lem:positivity} we find for $(q,\Gamma)$ a 
non-symmetric equilibrium stress matrix~$\Omega$ with the desired
positive entries. We now show that $\Omega$ has the stated rank.

First remove the $d+1$ rows and columns associated with
the exceptional  vertices to create a smaller matrix $\Omega'$.
Due to the sign pattern assumed in $\Omega$, as well as
property~(3) of any equilibrium stress matrix,
$\Omega'$ must be weakly diagonally dominant.

Let us call a vertex  \emph{EN} if it has an exceptional neighbor
and refer to its corresponding row in $\Omega'$ as EN.
Any EN row 
must be
strictly diagonally dominant (since at least one non-zero 
off-diagonal entry of $\Omega$ have been removed from this row).

Since all entries corresponding to edges are non-zero,
the irreducible components of $\Omega'$ correspond to vertex subsets that
remain connected after the exceptional vertices have been removed.
(An irreducible square matrix is one that is not similar via a permutation to
a block upper triangular matrix. Any square matrix has a unique
irreducible decomposition).

Each irreducible component of $\Omega'$ includes
such an EN row, thus $\Omega'$ must be full rank.  (See, e.g.,
\cite[Theorem~1.21]{varga2000matrix}.)

Since $\Omega'$ has co-rank 0, the co-rank of $\Omega$ must be at most
$d+1$. It is no less since any equilibrium stress matrix must have
a $(d+1)$-dimensional kernel spanned by the coordinates of $q$ and
the all-ones vector.
\end{proof}

\begin{proposition}
\label{prop:nss}
Let $\Gamma$ be a graph with at least $d+1$ vertices.
Suppose $\Gamma$ is $(d+1)$-connected, and $p$ is 
generic in $C^d(\Verts)$.  
Then there is a non-symmetric equilibrium stress matrix
$\Omega$ of $(p,\Gamma)$ with co-rank $d+1$.
\end{proposition}
\begin{proof}
From Lemma~\ref{lem:rubber}, there must exist a generic framework
$(q,\Gamma)$ in $C^d(\Gamma)$
that has the  convex containment property. From
Lemma~\ref{lem:convex-rank}, $(q,\Gamma)$ must have a non-symmetric
equilibrium stress matrix of co-rank $d+1$. Thus from
Proposition~\ref{prop:genrank}, any generic framework 
$(p,\Gamma)$ must
have such a matrix as well.
\end{proof}

\begin{figure}[t]
\center
\includegraphics[width=0.2\textwidth]{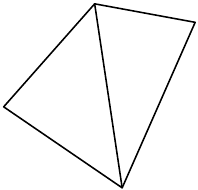}
\caption{\label{fig:book}
This framework in $\RR^2$ is not 3-connected but does have a 
non-symmetric stress matrix of high rank.
}
\end{figure}

See Figure~\ref{fig:book} for an example showing that the converse 
of Proposition~\ref{prop:nss} is not true.
Since the  upper (and lower) vertex in the framework has $3$ neighbors in 
affine general position, its position can be written as an affine combination
of these neighbors producing a non-zero, non-symmetric equillibrium 
stress matrix.
Any non-zero,
non-symmetric 
equillibrium stress matrix must have rank at least $1$
and
co-rank of no more than $3=d+1$.
Thus,
as an euillibrium stress matrix, it has co-rank of exactly $d+1$.
Meanwhile, this framework is not 3-connected.

Note that from the proof of Proposition~\ref{prop:genrank} it is
clear that if $\Gamma$ is $(d+1)$-connected, then \emph{almost
every} non-symmetric stress matrix for almost every $(p,\Gamma)$ 
in $C^d(\Gamma)$ will have co-rank $d+1$. 
Moreover, each row of such a non-symmetric stress
matrix of $p$ can be constructed independently from the
other rows, and we still expect to find this minimal co-rank.

\begin{proposition}
\label{prop:ar}
Let $\Gamma$ be a graph with at least $d+1$ vertices.
Suppose $(p,\Gamma)$, a framework in $C^d(\Gamma)$,
has a non-symmetric equilibrium stress matrix
$\Omega$ that has co-rank $d+1$.
Then $(p,\Gamma)$ is neighborhood affinely rigid in $\RR^d$.
\end{proposition}
\begin{proof}
$\Omega$ is a (not strong) affinity matrix of $(p,N(\Gamma))$
and so the proof follows that of the first direction
of Theorem~\ref{thm:ranktest}.
\end{proof}

\begin{proof}[Proof of Theorem~\ref{thm:dcon}]
The theorem now follows directly from 
Propositions~\ref{prop:nss} and~\ref{prop:ar}. 
If $\Gamma$ has less than $d+1$ vertices and is
$(d+1)$-connected, then it is a simplex and thus 
neighborhood affinely rigid for any configuration.
\end{proof}

\begin{figure}[t]
\center
\includegraphics[width=0.2\textwidth]{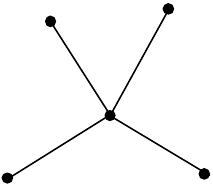}
\caption{\label{fig:spider}
This framework in $\RR^2$ 
does not have a non-symmetric 
equilibrium stress matrix of co-rank $d+1=3$, but is
(trivially) neighborhood affinely rigid.
}
\end{figure}

See Figure~\ref{fig:spider} for an example showing that the converse 
of Proposition~\ref{prop:ar}
is not true.
The framework is clearly 
neighborhood  affinely rigid since the central vertex is adjacent to
all of the other vertices. Meanwhile the outer 4 vertices have
only one neighbor and hence must have all zeros in their corresponding
rows of any
non-symmetric equilibrium stress matrix.
Thus, this framework  cannot have 
a  non-symmetric 
equilibrium stress matrix of co-rank $d+1=3$.

\begin{remark}
Generic global rigidity of a graph~$\Gamma$ in $\EE^d$
can be characterized either using the dimension of the kernel
of a single symmetric stress matrix of a generic framework $(p,
\Gamma)$ or
using the dimension of the \emph{shared symmetric
stress kernel} of a generic~$p$: the intersection of the kernels of
all stress matrices of~$p$~\cite{GHT10}.

By contrast, the analogous statement is not true in the affine
rigidity case.  By ``vertically concatenating'' a sufficient number of
non-symmetric equilibrium stress matrices of $(p,\Gamma)$, we can
create a strong affinity matrix for $(p,N(\Gamma))$.  The kernel of
the vertical concatenation will be the shared non-symmetric stress
kernel of $(p,\Gamma)$, and the dimension of this kernel
characterizes affine rigidity.
Since the converse of Proposition~\ref{prop:ar}
is not true, we see that 
neighborhood affine rigidity cannot in general be characterized by
the rank of one single  (say, generic)
non-symmetric equilibrium
stress matrix for $(p,\Gamma)$.
\end{remark}

Note that there is a different sufficiency condition for 
affine rigidity given by Zha and Zhang~\cite{ZZ09}. 
Their condition is complementary to
our condition (neither strictly stronger or weaker), and (like
trilateralization~\cite{eren2004rigidity}) 
is greedy in nature.  Their condition on generic
frameworks of a hypergraph
requires that for each pair of vertices  $s$ and $t$,
one can find a sequence of hyperedges starting with some hyperedge containing
$s$ and ending with some hyperedge containing $t$, 
such that for each pair $(i,j)$ of hyperedges in the sequence,
$h_i$ and $h_j$ share at least $d+1$ vertices.  When translated
to a neighborhood hypergraph $N(\Gamma)$,
it states that one can walk between any two vertices along edges, such
that for each pair  $(i,j)$ of vertices along the walk,
the neighborhoods of these vertices share at least $d+1$ vertices
in $\Gamma$.

Figure~\ref{fig:hcomb} shows a graph which clearly fails Zha and
Zhang's condition, but is 3-connected, showing that their condition
does not imply Theorem~\ref{thm:dcon}.  It is not hard to construct
examples in the opposite direction, as well.

\begin{figure}[t]
\center
\includegraphics[width=.5\textwidth]{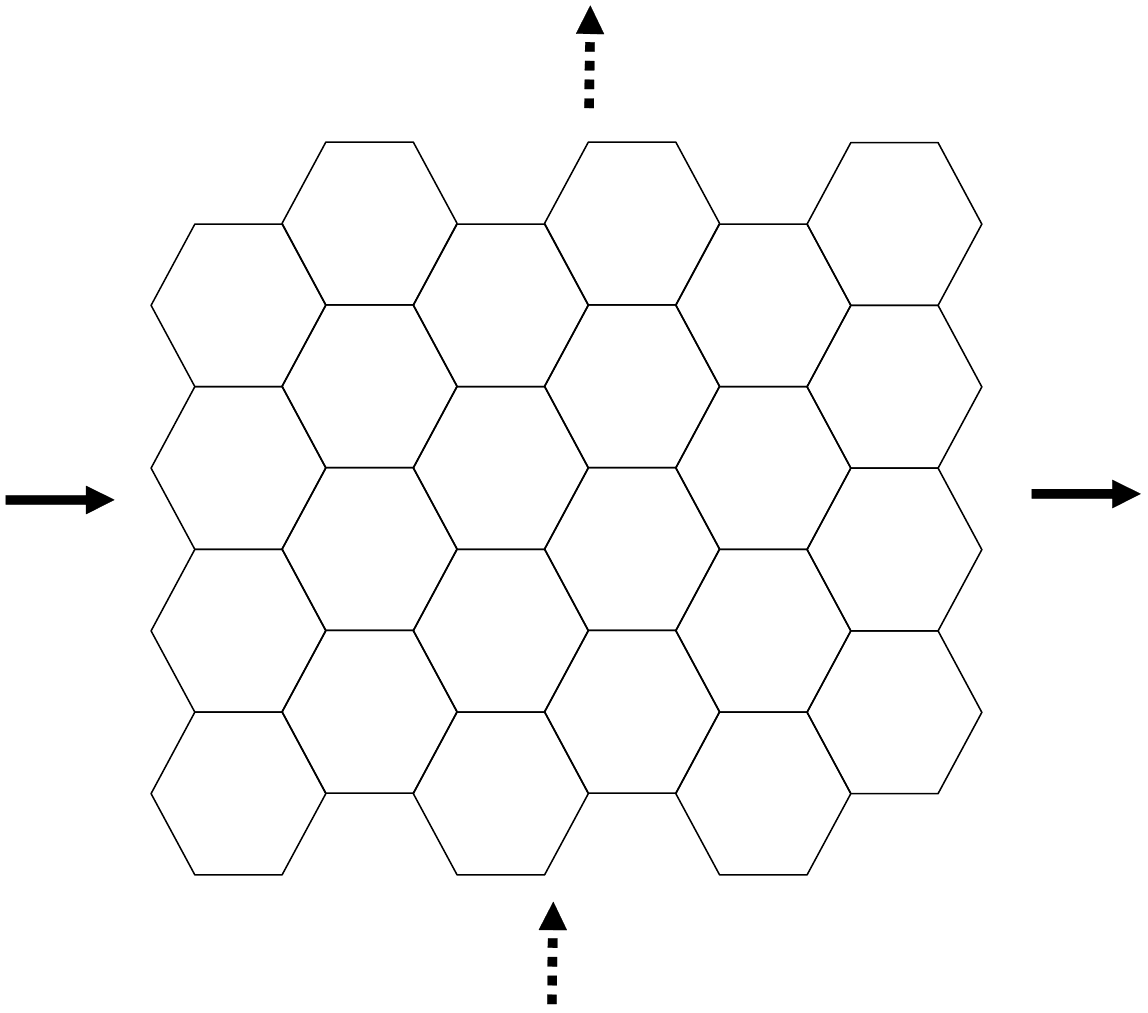}
\caption{\label{fig:hcomb}
A drawing of a hexagonal lattice on the torus.  (The vertices on the
top edge should be identified with those on the bottom and similarly
with the left and right, as indicated by the arrows.)
This graph is $3$-connected 
but its neighborhood hypergraph 
does not satisfy the sufficiency condition of Zha and
Zhang~\cite{ZZ09}, and its
squared graph is not a 2-trilateralization graph.
}
\end{figure}

We also have the following
corollary of Theorem~\ref{thm:dcon}

\begin{corollary}
Let $\Gamma$ be a graph.
If\/ $\Gamma$ is $(d+1)$-connected, then any generic framework of\/ $\Gamma^2$
in $\RR^d$ is universally rigid. 
\end{corollary}
\begin{proof}
If $\Gamma$ has at least $d+2$ vertices, then we can directly apply
Corollary~\ref{cor:ur}. Any graph with fewer vertices that 
is $(d+1)$-connected must be a simplex and be universally rigid
for all configurations.
\end{proof}

\begin{remark}
This corollary can also be proven without reference to affine rigidity and
Corollary~\ref{cor:ur}. 
In particular, Proposition~\ref{prop:nss} guarantees a non symmetric
maximal-rank
stress $\Omega$ for 
$(p,\Gamma)$, 
and then $\Omega^t\Omega$ is a symmetric, positive semi-definite,
maximal rank stress
for 
$(p,\Gamma^2)$. Universal rigidity then follows by a theorem of
Connelly~\cite{Connelly82:RigidityEnergy}. 
(See also~\cite{GT10:UGR}.)
\end{remark}

A manuscript by Cheung and Whiteley~\cite{cw} contains a variety of other 
results relating graph powers to global rigidity.

We wish to highlight this corollary since 
the only other known (to us) class of graphs that are universally rigid
for all generic configurations in $\RR^d$ are graphs that can be realized
greedily such as 
the $d$-trilateralization 
graphs 
(A $d$-trilateralization graph is one that can be obtained from a
complete graph by successively adding vertices, each connected to at
least $d+1$ old vertices)
and their generalizations (such as graphs formed by 
gluing together $d$-trilateralization graphs along $d+1$ vertices.

See Figure \ref{fig:hcomb} for an example 
of a framework  whose square is universally rigid by this
corollary but is not a trilateralization graph.

We also mention that a theorem of a related nature, 
showing a relationship between the connectivity of a graph
and global rigidity in the squared graph, has been described 
by Anderson et al.~\cite{anderson2009graphical}.


\section{Applications}

\subsection{Registration}

There are many applications where one has multiple views of some underlying
configuration, but it is not known how these views all fit together. 
We assume that these views share some points in common, and this
correspondence is known. (Of course in practice, 
establishing such a correspondence could in itself be a very
challenging problem.)
For example, in computer vision, one may have multiple uncalibrated
laser scans of
overlapping parts of some three-dimensional object.

In our setting we model all of the points as vertices, and
each of the views as a hyperedge. 
The geometry of the vertices in each hyperedge $h$ is given
up to some unknown transform~$T_h$ from a relevant class. 
The goal in registration then is
to \emph{realize} the entire hypergraph up to the relevant congruence class.

\paragraph{\bf Affine case:}
Suppose we wish to realize a framework $(p,\Theta)$ where 
we are given as input the geometry of each hyperedge $h$ up to
an affine transform $A_h$. Theorem~\ref{thm:ranktest} tells
us  that if $(p,\Theta)$ is affine rigid, then we can compute 
the realization just using linear algebra. 
In particular, we can use the data for each hyperedge to build
its associated rows in a strong affinity matrix. Then we can solve for its
kernel, giving us our answer $p$.

If our hypergraph $\Theta$ happens to be the neighborhood
 graph of an underlying graph~$\Gamma$,
then one could 
also  construct a (smaller)
non-symmetric equilibrium stress matrix~$\Omega$ for 
$(p,\Gamma)$.  
This is not guaranteed to work; 
even when $(p,\Gamma)$ is neighborhood affinely rigid in~$\RR^d$, 
the matrix~$\Omega$ may have co-rank larger than $d+1$.
But Theorem~\ref{thm:dcon} states that if $\Gamma$ is
$(d+1)$-connected, this method will indeed
work for almost every $p$ in $\RR^d$
(and, in fact, using almost any non-symmetric equilibrium stress matrix for 
$(p,\Gamma)$).

\paragraph{\bf Euclidean case:}

The Euclidean framework registration problem
is perhaps more natural and common.

When $(p,\Theta)$ is globally  rigid in $\RR^d$, this problem is well 
posed, 
but it 
is general hard to solve, as the graph case includes the
graph realization problem
which is strongly NP-HARD~\cite{Saxe79:EmbedGraphsNP}.

When $(p,\Theta)$ is, in fact, also universally rigid there is an
efficient algorithm:
we can solve the Euclidean registration problem 
using semi-definite programming.
One simply sets up the program that 
looks for the Gram matrix of an embedding of the vertices in $\RR^v$
(a semi-definite
constraint on a Gram matrix)
subject to the length constraints (linear constraints on 
the Gram matrix)~\cite{linial1995geometry}.
Due to 
universal rigidity, 
one does not need to explicitly enforce 
the (non-convex) constraint that the embedding 
have a $d$-dimensional affine span~\cite{so2007theory}.

When $(p,\Theta)$ is, furthermore, affinely rigid then
we can
solve the Euclidean registration problem using linear algebra.  We can
simply reduce this problem to an affine registration problem above,
and find $p$ using the kernel vectors of an affinity matrix. This
determines $p$ up to some global affine transform.  Moreover,
for $(p,\Theta)$ that is generically globally rigid in $\RR^d$, we
can solve a second (least squares) linear system 
to remove the
unwanted global affine transform, leaving just the unknown global
Euclidean transform (see Appendix~\ref{sec:rta}). This approach
is morally the same
affine relaxation 
used in the initialization step of the registration method 
of Krishnan et al.~\cite{krishnan2005global} (though in their case,
they think of 
the inter-patch transforms as the unknown variables instead of the 
point positions).

\subsection{Global embeddings from edge lengths}

Similar approaches have been applied to the (NP-HARD) problem of
solving for the 
framework of a graph  given its edge lengths. In these approaches one
first attempts to find local $d$-dimensional embeddings for each
one-ring (a vertex and its neighbors)
of the framework
up to an unknown local Euclidean transform. This step alone
is NP-HARD and can fail. But assuming this step is (approximately)
successful one can reduce the rest of the problem to the Euclidean
registration problem above.  

In the As-Affine-As-Possible (AAAP) method~\cite{koren-patchwork,zhang2009rigid}, this was done using what is essentially 
a strong affinity matrix.  In the Locally-Rigid-Embedding (LRE)
method~\cite{singer2008remark} this was done using a non-symmetric equilibrium stress matrix. Both approaches then 
removed the global affine transform using the 
least squares linear system described 
in  Appendix~\ref{sec:rta}.

\subsection{Manifold learning}

Many of the ideas of affine rigidity
first appeared in the context of manifold learning.
Suppose one has $d$-dimensional smooth manifold $\cM$ 
which is a topological $d$-ball  embedded
in a larger $D$-dimensional space $\RR^D$.  
Also suppose that one has  a set $\Verts$ of sample vertices on the
manifold. 
In manifold learning, one first
connects nearby samples to form a proximity graph $\Gamma$.
One then looks for a framework $(p,\Gamma)$ 
of this graph in~$\RR^d$ that
in some way that preserves some of 
the geometric relations of the points in~$\RR^D$. This is  used to
represent a parametrization of $\cM$.

To compute the coordinates $p$,
the Locally Linear Embedding (LLE) method~\cite{roweis2000nonlinear} 
builds a matrix $\Omega$ with structure similar to
a non-symmetric equilibrium stress matrix. In particular, row $i$
encodes one affine relation between vertex $i$ and its neighbors
in~$\RR^D$. Then (after ignoring the all ones vector) the 
smallest $d$ eigenvectors of 
$\Omega^t \Omega$ are used to form
 the coordinates of $p$ in
$\RR^d$. Unfortunately, since the original embedding is in $\RR^D$, 
for a graph of high enough valence and assuming no noise,
$\Omega$ must have a kernel of size at least $D+1$, which is much larger
than $d+1$.  Thus it is not clear how the numerically smallest $d+1$
eigenvectors will behave. The paper suggests to add 
an additional regularization
term, possibly to address this issue.

A follow up to the
LLE paper~\cite{saul2003think}
describes a PCA-LLE variant where a
$d$-dimensional local PCA is computed to ``flatten'' each one-ring
before calculating its corresponding
row in the matrix $\Omega$. Thus $\Omega$ is designed
to represent $d$-dimensional
affine relations between the points.

The Local-Tangent-Space-Alignment (LTSA) method~\cite{ZZ04} is an interesting 
variant of PCA-LLE.  In this
method, a $v \times v$ matrix $N$ is formed that is the Hessian of a
quadratic energy. 
Thus this matrix plays the role
of a strong affinity matrix. It is in this context that 
Zha and Zhang investigated the rank of this matrix and 
affine rigidity~\cite{ZZ09}. 

In all of these methods, an understanding of affine rigidity is
important. In particular it tells us what the rank of the
computed matrix would be if the original $d$-dimensional manifold was
in fact embedded in $\RR^d$. For example, if such a framework
was not affinely rigid in $\RR^d$, then the
kernel would  be too big, and we would not expect a manifold
learning technique to succeed.
However, in manifold learning
the embedding 
has an affine span greater than $d$ and the analysis becomes more
difficult. 
The kernel of a strong $D$-dimensional affinity matrix is too large,
while the kernel of a strong affinity matrix for the locally flattened
configurations
contains only the 
all-ones vector but 
it is hoped that the numerically next smallest
$d$ eigenvectors are somehow geometrically meaningful.
For an analysis along these lines, see~\cite{SmithHZ08}.


\appendix
\section{Matrix rank}
For completeness, we recall the necessary material for determining
the generic matrix rank. This material is standard.  For a 
more detailed
treatment, see, e.g., \cite[Section~5]{GHT10}.

We will consider the general setting 
where  there is a set of linear
constraints that must be satisfied by a vector
$m \in \RR^n$.
The entries of~$m$ are then arranged in some fixed manner
as the entries of a matrix $M$, whose rank we wish to understand. 
The linear constraints are described by a 
constraint matrix with $n$ columns:
$C(\prho)$.
Each of the entries
of $C(\prho)$ is defined by 
some polynomial function, with coefficients in~$\QQ$,
of the coordinates of an input
configuration~$\prho$. 
We wish to study the behavior of the rank of $M$ as one changes $\prho$.

We apply this 
in the proof of Corollary~\ref{cor:arGen} where the constraints $C$
specify that
the matrix~$M$ is an affinity matrix for $(\prho,\Theta)$
and in the proof of Proposition~\ref{prop:nss}
where the constraints~$C$ specify that
the matrix $M$ is a non-symmetric equilibrium stress matrix
for $(\prho,\Gamma)$.

\begin{proposition}
\label{prop:genrank} 

Suppose that for some generic
$\prho$, there is a matrix $M$ of rank $s$ consistent with $C(\prho)
m =0$.

Then for all generic $\prho$, there is some matrix $M$ of rank
$\geq s$ consistent with $C(\prho) m = 0$.
\end{proposition}

\begin{proof}

To prove this proposition we first need the following lemma.

\begin{lemma}\label{lem:poly-mat}
Let $M(\pi)$ be a matrix whose entries are polynomial functions
with rational coefficients in the
variables $\pi \in \RR^n$.
Let $r$ be a rank achieved by some $M(\pi_0)$.
Then $\rank (M(\pi))\geq r$
for all points $\pi$ that are generic in $\RR^n$.
\end{lemma}

\begin{proof}
The rank of the $M(\pi)$  is
less than~$r$ iff the determinants of all of the
$r\times r$ submatrices
vanish.
Let $\pi_0\in\RR^n$ be a choice of parameters so $M(\pi_0)$ has
rank~$r$.  Then there is an $r \times r$ submatrix $T(\pi_0)$ of
$M(\pi_0)$ with non-zero determinant.  Thus $\det(T(\pi))$ is
a non-zero polynomial of~$\pi$.
For any $\pi$ with $\rank(M(\pi)) < r$, this
determinant must vanish.
Thus, any such $\pi$ cannot be generic.
\end{proof}

Next we recall that for a non-singular $n \times n$ matrix~$\hC$,
\begin{equation}\adj(\hC) = \det(\hC)\hC^{-1},\end{equation}
where $\adj \hC$ is the \emph{adjugate
 matrix} of $\hC$, the conjugate of the cofactor matrix of $\hC$.
In particular,
$\adj{\hC}$ is a polynomial in~$\hC$.

For a given $\prho$, let $t(\prho)$ be the rank of $C(\prho)$.
Let $t:= \max_{\prho} t(\prho)$.  By Lemma~\ref{lem:poly-mat} this
maximum is obtained for generic~$\prho$.

For each $\prho$
we add a set~$H$ of $n-t$ additional rows
to $C(p)$ to obtain
a  matrix $C(\prho,H)$, and
determine  $m$
by solving the  linear system $C(\prho,H) m = b$
where $b \in \RR^n$ is a vector of all zeros except for a
single~$1$ in one of the positions of a row in~$H$ (if there are any
rows in~$H$).
This $m$ is
then converted to a matrix $M(\prho,H)$.
$M(\prho,H)$ is well-defined iff this linear system
has a unique solution, i.e., iff $C(\prho,H)$ has rank~$n$.  Note that
this happens for generic $\prho$ and~$H$.

Let $\prho_0$ be generic and have a compatible matrix $M_0$
with rank~$s$, as in the hypotheses of the proposition.
Find a set~$H_0$ of additional rows so that $C(\prho_0,H_0)$ has
rank~$n$ and $C(\prho_0,H_0)m_0 = b$.  Let $\hC(\prho,H)$ be an $n
\times n$ submatrix of $C(\prho,H)$ so that $\hC(\prho_0,H_0)$ is
invertible.  ($\hC$ necessarily uses $t$ rows from~$C(\prho)$ and
all rows of~$H$.)
Define
$\hb$ similarly, let $\tm(\prho,H) := \adj(\hC)\hb$, and let
$\tM(\prho,H)$ be
the associated matrix.

By Lemma~\ref{lem:poly-mat}, the rank of $\hC(\prho,H)$ is equal to its
maximum value, $n$, at all points $(\prho,H)$ that are not zeros of a
polynomial $P_1(\prho, H) := \det \hC(\prho,H)$.
Moreover, when $P_1(\prho, H) \ne 0$, the linear equation defining
$M(\prho,H)$ has a unique solution and the adjugate
matrix $\tM(\prho, H)$ is a scalar multiple of $M(\prho,H)$.
In
particular
we have assumed $(\prho_0, H_0)$
is not a zero of $P_1$ and thus
$\tM(\prho_0, H_0)$ has rank~$s$.  By
Lemma~\ref{lem:poly-mat} again, the rank of $\tM(\prho,H)$
is less than $s$ only
at the zeros of a non-zero polynomial $P_2(\prho, H)$.

For any generic $\prho$, there must be some generic point
$(\prho,H)$. Such a generic $(\prho,H)$ cannot be a zero of $P_1$ or
$P_2$ and thus $\tM(\prho,H)$ and $M(\prho,H)$ must have rank no
less than~$s$.
\end{proof}


\section{Removing the Affine Transform}
\label{sec:rta}

Suppose one has solved for $\psigma$ -- 
a configuration  in $\RR^d$ -- up to
an unknown global affine transform~$A$
of the true configuration $\prho$: $\prho = A (\psigma)$.
Given a set of edge lengths 
for $\prho$, it is possible 
to compute $A$ up to an unknown global Euclidean transform.
This approach was described by Singer~\cite{singer2008remark}.

In particular, 
let $L$ be a $d \times d$ matrix representing the 
linear portion of $A$ and let $G:=L^TL$.  
Now consider the following set of linear
equations (in the $\frac{d(d+1)}{2}$ unknowns of $G$):
For each pair of vertices $i,j$, whose edge lengths are known, 
we require
\begin{equation}
\label{coninf}
  (\psigma(i)-\psigma(j))^T G(\psigma(i)-
\psigma(j))= 
  (\prho(i)-\prho(j))^T (\prho(i)-
\prho(j)) 
\end{equation}
(Since we have more constraints than unknowns, for numerical purposes
we would solve this as a least squares linear system in the unknown $G$.)

The only remaining concern is whether this system has more than
one solution.  
The solution to Equation~\eqref{coninf} will be unique
as long as our set of edges with known lengths do not lie
on a conic at infinity. 

Fortunately, we can conclude from Proposition 4.3
of~\cite{Connelly05:GenericGlobalRigidity} that 
if our known lengths form a graph $B(\Theta)$
with minimal valence at least $d$
and $\prho$ is generic, then the edges  do not lie on a conic at infinity. 
This property holds
for any  hypergraph that is generically globally rigid in $\EE^d$.

Using Cholesky decomposition on $G$ then yields $L$ up to a global
Euclidean transform.


\bibliographystyle{hacm}
\bibliography{lle,graphs}
\end{document}